\documentclass[12pt,a4paper]{article}
\usepackage[utf8]{inputenc}
\usepackage{natbib}
\usepackage[margin=1in]{geometry}

\usepackage{amsthm, amssymb, amsmath, appendix, physics, url}
\usepackage{bm}
\usepackage{setspace}

\usepackage{mathtools}
\mathtoolsset{showonlyrefs=true}

\usepackage{fourier}
\usepackage[T1]{fontenc}

\DeclareMathAlphabet{\mathcal}{OMS}{cmsy}{m}{n}

\usepackage[mathlines]{lineno}
\renewcommand{\thesection}{\Roman{section}}
\renewcommand{\thesubsection}{\Alph{subsection}}

\usepackage{titlesec}

\titleformat{\section}
  {\normalsize\bfseries\centering} 
  {\thesection.} 
  {0.5em} 
  {} 

\titleformat{\subsection}
  {\normalfont\centering} 
  {\thesubsection.}    
  {0.5em}              
  {\itshape}    

\titleformat{\subsubsection}
  {\normalfont\centering} 
  {\thesubsubsection}    
  {0.5em}              
  {}    

\titleformat{name=\section, numberless}
  {\normalsize\centering\scshape}
  {}
  {0em}
  {}

\theoremstyle{definition}

\theoremstyle{plain}
\newtheorem{theorem}{Theorem}
\newtheorem{lemma}{Lemma}
\newtheorem{axiom}{Axiom}
\newtheorem{proposition}{Proposition}
\newtheorem{corollary}{Corollary}

\usepackage{tikz}
\usetikzlibrary{intersections, calc, arrows.meta}



\newcommand{\pstar}{$\raisebox{0.2ex}{*}$}

\title{\fontsize{19pt}{18pt}\selectfont Reservation of Judgment and Robust Collective Decisions}
\author{Leo Kurata and Kensei Nakamura}
\date{This version: \today}

\begin{document}

% \onehalfspacing
\sloppy

\maketitle

\begingroup
\def\thefootnote{}
\footnotetext{The authors are grateful to Tsuyoshi Adachi, Noriaki Kiguchi, Kaname Miyagishima, Satoshi Nakada, Hendrik Rommeswinkel, Koichi Tadenuma, Norio Takeoka, Gerelt Tserenjigmid, and Shohei Yanagita for their helpful comments. 
This research is presented at Decision Theory Workshop (Hiroshima University of Economics) and the seminar at Tokyo University of Science. 
This research is financially supported by KAKENHI (Nos.~25KJ2146 \& 25KJ1298). }
\footnotetext{Kurata: Graduate School of Economics, Waseda University, Tokyo, Japan (e-mail: \protect\url{leo.kurata.ac@gmail.com}). Nakamura: Graduate School of Economics, Hitotsubashi University, Tokyo, Japan (e-mail: \protect\url{kensei.nakamura.econ@gmail.com}).}
\endgroup

\begin{abstract}
This paper studies preference aggregation under ambiguity when agents have incomplete preference relations due to imprecise beliefs. We introduce the ``dual'' of the Pareto principle, which respects unanimity among individuals, including those with unexpressed opinions. Our first theorem shows that, in most cases, this principle leads to a dictatorial rule in taste aggregation. We argue that this stems from the problem of spurious unanimity, even when the individuals have the same prior set. By weakening the above principle to avoid respecting spurious unanimity, the second theorem characterizes novel belief-aggregation rules, under which society does not discard any combination of plausible priors. 
\vspace{2mm}
\end{abstract}

%%%%%%%%%%%%%%%%%%%%%%%%%%%%%%%%%%%%
\section{Introduction}
%%%%%%%%%%%%%%%%%%%%%%%%%%%%%%%%%%%%
Ambiguity---defined as the absence of objective, precise information about the true probability distribution---is ubiquitous in economic problems. 
Agents in society with ambiguity form their own tastes and predictions and make decisions using them. 
Then, they do not necessarily have a single, precise prior (cf.~\citealp{ellsberg1961risk}) and rather, often think multiple probability laws to be plausible. 
In these situations, how should a social decision maker, or \textit{society}, predict the future state and evaluate each alternative, such as public policies, while respecting individual values? 
This question is practically important when people hold different interests and beliefs but still need to act collectively, as in the problem of climate change. 

This paper studies collective decision-making rules respecting unanimity within the framework of \citet{danan2016robust}, in which each agent's preference relation is \textit{incomplete} due to the multiplicity of prior distributions \`{a} la \citet{bewley2002knightian}. 
When comparing two prospects, the classical criterion known as \textit{the Pareto principle} requires that
\vspace{2mm}
\begin{quote}
\centering
\begin{minipage}{0.75\linewidth}
\raggedright
    if all individuals evaluate X to be weakly better than Y, \\
    then society should reach the same conclusion.
\end{minipage}
\end{quote}
\vspace{3mm}
This condition has been widely studied in economics, especially in the context of collective decision making (e.g., \citealp{Arrow1951-ARRIVA,harsanyi1955cardinal}). 
This principle requires that society should respect unanimity reached among individuals who are sufficiently confident to express their value judgments. 
However, in many real-world situations, such as elections, a considerable number of individuals refrain from expressing their judgment since they lack confidence in the comparisons of the alternatives. 
While the standard Pareto principle ignores agreements including unexpressed opinions, the reservation of judgment conveys meaningful information for making collective decisions. It signals that those individuals have some reasons for withholding judgment that one alternative is more desirable than another. 
Motivated by this, we study a dual of the standard Pareto principle, which we call \textbf{\textit{the Pareto{\pstar} principle}}, as follows: 
\vspace{2mm}
\begin{quote}
\centering
\begin{minipage}{0.75\linewidth}
\raggedright
    If all individuals evaluate X to be \textbf{not} weakly better than Y, \\
    then society should reach the same conclusion.
\end{minipage}
\end{quote}
\vspace{3mm}

This paper provides two aggregation theorems about the Pareto{\pstar} principle. 
First, we show that imposing the Pareto{\pstar} principle  results in a severe impossibility.
It has been known that under ambiguity, the Pareto principle is incompatible with respecting all individuals' tastes, but it can be satisfied by the utilitarian aggregation rules only when all individuals have a common prior distribution \citep{hylland1979impossibility,seidenfeld1989shared,mongin1995consistent,mongin1998paradox,chambers2006preference,mongin2015ranking,zuber2016harsanyi}.
Furthermore, \citet{danan2016robust} showed that when aggregating incomplete multiprior preferences based on the Pareto principle, society can take all individuals' tastes into account only if at least one probability distribution is shared by these individuals. 
In contrast, our first theorem claims that \textbf{even when all individuals have a common prior set, the Pareto{\pstar} principle is incompatible with respecting their tastes.}%
\footnote{Aggregation of imprecise beliefs has been considered in the literature. 
\citet{gajdos2008representation}, \citet{mongin2015ranking}, and \citet{zuber2016harsanyi} showed that it is difficult to avoid dictatorship under the Pareto principle when aggregating general utility functions. 
On the other hand, our first result shows that severe impossibility also holds in the aggregation of incomplete preference relations by considering the dual of the Pareto principle, which cannot be obtained from the standard one (cf.~\citealp{danan2016robust,nakamura2025impossible}).}

Impossibility theorems in the literature on preference aggregation under ambiguity are rooted in \textit{spurious unanimity}: If individuals have different tastes and beliefs, then the double disagreements may cancel each other out and produce a merely superficial consensus. We argue that our severe impossibility under the Pareto{\pstar} principle also stems from spurious unanimity, even when all individuals have a common belief set. Note that an individual concludes that ``X is \textbf{not} weakly better than Y'' if he or she has at least one prior under which Y is better than X. 
Therefore, in cases of conflicting tastes, individuals with common multiple priors can reach agreement on these evaluations by relying on different priors in the belief set. Respecting these superficial consensuses among individuals leads to the unacceptable consequences that some individuals' tastes must be ignored.

Second, we examine a weakening of the Pareto{\pstar} principle, as in \cite{danan2016robust}. 
They introduced \textit{the common-taste Pareto principle} that requires society to respect unanimity if there is no disagreement in tastes over relevant outcomes.
Our new condition, which we call \textbf{\textit{the common-taste Pareto{\pstar} principle}}, is the counterpart of their axiom.
\citet{danan2016robust} showed that under the common-taste Pareto principle, individual tastes are aggregated in a utilitarian way, while the social belief set must be included in the convex hull of the union of the individual belief sets. 
The latter part can be restated as each prior in the social belief set must be a weighted sum of some combination of individual beliefs. 
This is a natural extension of the weighted-average social beliefs characterized in the single-prior model \citep{gilboa2004utilitarian,billot2021utilitarian}.
In contrast, our second theorem shows that the common-taste Pareto{\pstar} principle has a dual, complementary implication for the belief-aggregation part: 
\textbf{For each combination of individual prior distributions, there must exist a social belief constructed as a weighted sum of the individual ones.} 
In other words, society is prohibited from discarding any combination of individual beliefs. 
The class of rules characterized in the second result is another natural extension of belief averaging.

Furthermore, our aggregation rules have desirable properties, which are not necessarily satisfied by the rules in \citet{danan2016robust}. Suppose that some priors are plausible for all individuals. Then, since the ``true'' probability law is more likely to be in the set of priors that all individuals think plausible, it would be natural to expect that the social belief set includes them. However, under the aggregation rules characterized in \citet{danan2016robust}, it is permissible for society to exclude (a part of) those priors. In contrast, our aggregation rules can ensure that society has a belief set including these shared priors. In addition, combined with the common-taste Pareto principle, whenever individuals have a common set, the social belief must coincide with the individual one.

In summary, this paper offers new insights into preference aggregation under ambiguity by incorporating reservations of judgment---an aspect often neglected in the existing literature. 
Our first theorem establishes a more severe impossibility result, emphasizing the importance of agreement on beliefs for respecting all individuals' preferences. 
Moreover, our second theorem proposes novel aggregation rules satisfying several desirable properties that are overlooked unless reservations of judgment are taken into account. 
This theorem provides more normatively appealing ways of constructing a social belief set within the class of rules compatible with weak versions of the standard Pareto principle, such as the common-taste Pareto principle.

This paper is organized as follows: 
Section \ref{sec:example} provides an illustrative example to explain the key ideas in our results. 
Section \ref{sec:model} introduces the formal model. Section \ref{sec:main} provides the two main theorems---the Pareto{\pstar} principle yields a severe impossibility, but the common-taste version characterizes new aggregation rules. 
Section \ref{sec:dis} offers another characterization of the aggregation rules obtained in our second theorem and discusses the applicability of our rules to the aggregation of ambiguity preferences studied in the literature.

%%%%%%%%%%%%%%%%%%%%%%%%%%%%%%%%%%%%
\section{Illustrative Example}
\label{sec:example}
%%%%%%%%%%%%%%%%%%%%%%%%%%%%%%%%%%%%

To illustrate that the Pareto{\pstar} principle yields unacceptable consequences and to reveal that this is caused by respecting spurious unanimity even under belief agreement, we consider the following simple example.
Suppose two politicians, Ann and Bob, deliberate on which is socially better between a status quo policy and a reform policy.  
For simplicity, we assume that the status quo policy yields unambiguous consequences, and both Ann and Bob will obtain utility 0 with certainty. 
On the other hand, the outcome of the reform policy is largely affected by social states, such as international affairs, that become apparent after the policy is implemented. 
There are two states, $s_A$ or $s_B$. 
Then, a prior distribution $p = (p_A, p_B)$ can be identified with its probability $p_A$ assigned to $s_A$. 
We first consider the case where both individuals have the same set $[0.2, 0.8]$ of plausible priors. 
According to \citeauthor{bewley2002knightian}'s (\citeyear{bewley2002knightian}) model, each agent concludes that one policy is weakly better than the other if the expected utility of the former is equal to or higher than that of the latter for all possible probabilities. 

\begin{figure}
    \centering
    \includegraphics[width=0.8\linewidth]{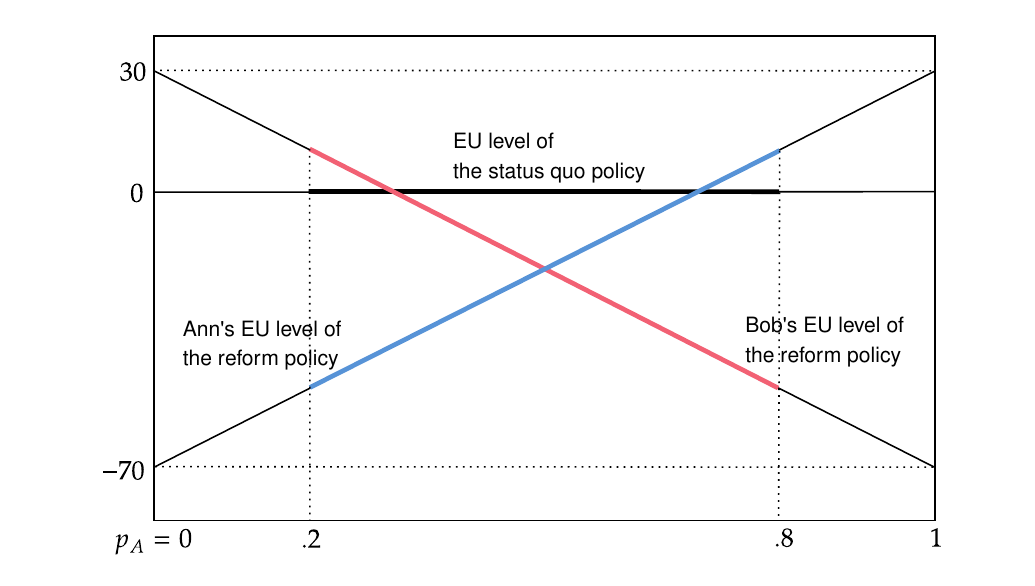}
    \caption{Spurious unanimity under a common belief set}
    \label{fig:example1}
\end{figure}

The politicians have conflicting political ideologies. 
Under the reform policy, Ann 
will obtain utility 30 in state $s_A$ and -70 in state $s_B$, while Bob will achieve utility -70 in state $s_A$ and 30 in state $s_B$. 
Figure \ref{fig:example1} plots both politicians' expected utility levels of the two policies for each plausible prior. 
The Pareto principle requires that if the expected utility level is greater than or equal to $0$ for every $p_A \in [0.2, 0.8]$ and for each politician, then the reform policy should be socially at least as good as the status quo policy. 
However, since this presumption does not hold, the Pareto principle does not offer any restriction on the collective preference relation. 
On the other hand, the Pareto{\pstar} principle works if the expected utility level is not always below $0$ on $[0.2, 0.8]$ for each politician.  
As illustrated in Figure \ref{fig:example1}, the condition is satisfied in this case, so the status quo is not weakly better than the reform policy at the social level. 

However, it is questionable whether such unanimity is really meaningful. 
Ann concludes that the status quo is not weakly better than the reform policy using priors that assign higher probabilities to state $s_A$ (e.g., probability $p_A = 0.8$), and Bob does so using priors that assign higher probabilities to state $s_B$ (e.g., probability $p_A = 0.2$). Thus, they reach a consensus through disagreements in both tastes and beliefs; that is, this agreement is just spurious unanimity, even though both politicians have a common belief set.

\begin{figure}
    \centering
    \includegraphics[width=0.8\linewidth]{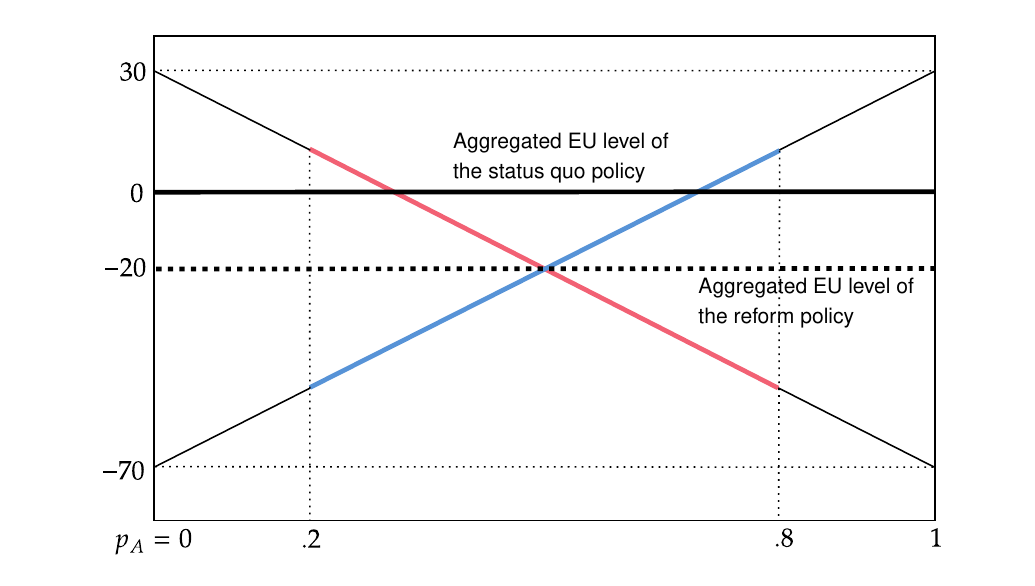}
    \caption{A counterexample}
    \label{fig:example1social}
\end{figure}

To see why the Pareto{\pstar} principle leads to an impossibility, consider the simplest taste-aggregation rule---equal-weight averaging. 
Figure \ref{fig:example1social} adds the graphs of expected aggregated utility levels to Figure \ref{fig:example1}. 
Regardless of the social prior set, the expected utility level of the status quo policy is always higher than that of the reform policy across the set. 
Therefore, the status quo policy should be better than the other. 
However, the Pareto{\pstar} principle gives the contrary evaluation. 
The impossibility result can be established because a similar argument holds for any taste-aggregation rule that is compatible with the Pareto{\pstar} principle.

Motivated by this argument, our second theorem restricts attention to agreements on pairs of acts for which there is no taste disagreement (i.e., it adopts the common-taste Pareto{\pstar} principle).
This restriction allows us to avoid the impossibility that arises from spurious unanimity. 
To see the intuition behind the second theorem, we then assume that both Ann and Bob will obtain utility 30 in state $s_A$ and -70 in state $s_B$ from the reform policy, but they have different beliefs.
Ann thinks state $s_A$ will be realized with probability ranging from 0.6 to 0.8, and Bob does so with probability ranging from 0.2 to 0.3 (see Figure \ref{fig:example2}). 
Then, since the graph of expected utility level of the reform policy is not above 0 for each politician, the common-taste Pareto{\pstar} principle requires that the reform policy should not be at least as good as the status quo policy for the society. 

\begin{figure}
    \centering
    \includegraphics[width=0.8\linewidth]{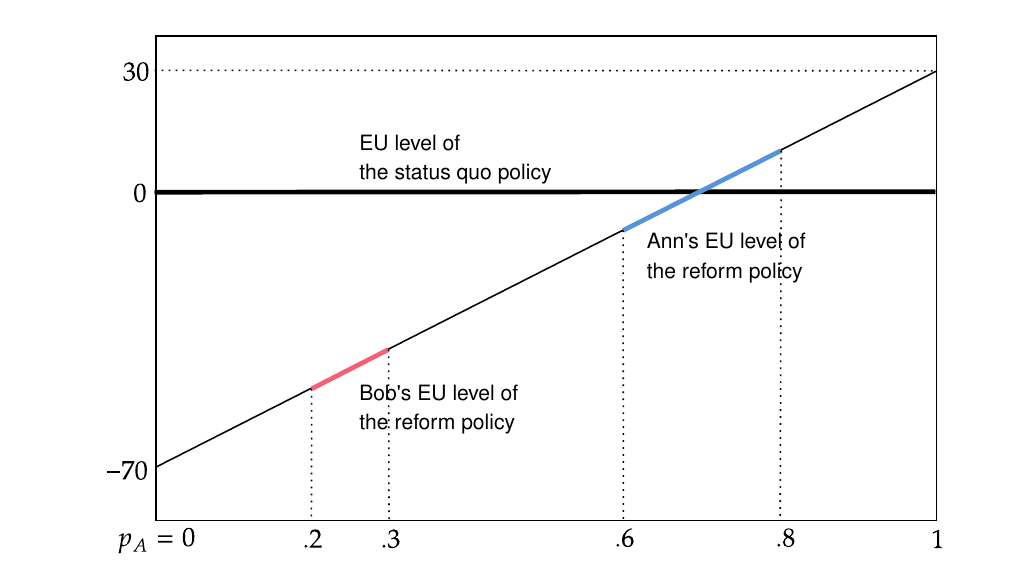}
    
    \caption{The common-taste Pareto{\pstar} principle and heterogeneous belief sets}
    \label{fig:example2}
\end{figure}

\citet{danan2016robust} showed that, under the common-taste Pareto principle, the collective belief set must be in the convex hull of individual belief sets, i.e., $[0.2, 0.8]$ in our example. 
Suppose that the collective belief is the singleton $\{0.8\}$. 
Then the reform policy becomes socially better than the status quo, so a contradiction to the common-taste Pareto{\pstar} principle arises. 
Therefore, the convex-hull approach cannot ensure satisfying the common-taste Pareto{\pstar} principle. 

Note that we can interpret the agreement respected by the common-taste Pareto{\pstar} principle as Ann and Bob each arriving at the same conclusion using different priors, for instance, priors that assign 0.6 and 0.3 to $s_A$, respectively. 
If there exists a prior $p_A$ between 0.3 and 0.6 in the collective belief set, then the society can also attain the conclusion that the reform policy should not be at least as desirable as the status quo policy. 
Our second theorem generalizes this observation to satisfy the common-taste Pareto{\pstar} principle in any case: 
Every combination of individual priors must be used to construct the social belief set. 
For two alternatives over which individuals have no taste disagreement, 
if they unanimously conclude that one is not better than the other, then each individual has a prior that supports the conclusion.
Since at least one social prior is constructed from this combination of priors under our aggregation rules, 
the social evaluation aligns with the agreement among individuals, that is,
the common-taste Pareto{\pstar} principle can be ensured. 
Our second theorem formally establishes this connection. 

%%%%%%%%%%%%%%%%%%%%%%%%%%%%%%%%%%%%
\section{Model}
\label{sec:model}
%%%%%%%%%%%%%%%%%%%%%%%%%%%%%%%%%%%%

We study the model examined in \citet{danan2016robust}. 
For any finite set $Y$, let $\Delta(Y)$ denote the set $\qty{ \alpha\in \mathbb{R}_+^Y \mid \sum_{y \in Y} \alpha (y) = 1}$. 
Let $S = \{ 1, 2, \ldots, m\}$ be a set of \textit{states} with $ m\geq 2$. 
Then, $\Delta (S)$ is the set of probability distributions over states. 
Let $X$ be a set of \textit{outcomes}. 
We assume that $X$ is a nonempty convex subset of a real vector space (e.g., the set of lotteries over finite deterministic prizes). 
An \textit{act} is a function from $S$ to $X$. 
The set of acts is denoted by $\mathcal{F}$. 
For $f,g\in \mathcal{F}$ and $\alpha \in [0,1]$, define $\alpha f + (1-\alpha)g \in\mathcal{F}$ as the act $h$ such that $h(s) = \alpha f(s) + (1-\alpha) g(s) $ for all $s\in S$. 
With an abuse of notation, 
we identify the outcome $x$ with the constant act $f$ such that $f(s) = x$ for all $s\in S$. 

Let $N = \{ 1,2, \ldots, n\}$ be a set of individuals with $n\geq 2$. 
The set $\Delta(N)$ represents the set of weights over the individuals. For $\gamma\in \Delta(N)$ and $i\in N$, we denote $\gamma (i)$ by $\gamma_i$ for notational simplicity. 
Society is indexed by $0$. 
Each agent $i \in N\cup\{ 0 \}$ has a preference relation $\succsim_i$ over $\mathcal{F}$. For all $f,g\in \mathcal{F}$, $f\succsim_i g$ means that agent $i$ thinks $f$ to be at least as good as $g$. 
The symmetric and asymmetric parts are denoted by $\sim_i$ and $\succ_i$, respectively. 

We assume that each agent $i \in N\cup \{ 0\}$ has a Bewley preference relation \citep{bewley2002knightian}; that is, there exists a pair $\qty(u_i, P_i)$ of a real-valued nonconstant affine function on $X$ and a nonempty compact convex subset of $\Delta(S)$ such that for all $f,g\in \mathcal{F}$, 
\begin{equation*}
    f\succsim_i g 
    \iff
    \qty[ ~ \sum_{s\in S} p(s) u_i(f(s)) \geq \sum_{s\in S} p(s) u_i(g(s))  ~~~\text{for all } p\in P_i~ ]. 
\end{equation*}
The function $u_i$ represents agent $i$'s tastes on outcomes, and the set $P_i$ is the set of probability distributions that agent $i$ considers plausible. 
According to this model, the agent thinks $f$ to be weakly better than $g$ if the expected utility of $f$ is not lower than $g$ for each plausible prior in $P_i$. 
Note that if the set $P_i$ is a singleton i.e., $P_i = \{p_i\}$ for some $p_i\in \Delta(S)$, then it becomes an \textit{SEU preference relation}: For all $f, g\in \mathcal{F}$, $f\succsim_i g$ if and only if $\sum_{s\in S} p_i(s) u_i(f(s)) \geq \sum_{s\in S} p_i(s) u_i(g(s))$. 

We introduce two assumptions for preference profiles. 
We say that the profile $\qty(\succsim_i)_{i\in N}$ satisfies \textit{c-minimal agreement} if there exist $x^*, x_*\in X$ such that for all $i\in N$, $x^* \succ_i x_*$. 
This only requires that there exists at least one pair of outcomes where all individuals agree upon the strict ranking.
Also, we say that $\qty(\succsim_i)_{i\in N}$ satisfies \textit{c-diversity} if for each $i\in N$, there exist $x^{i *}, x_{i *} \in X$ such that $x^{i *}\succ_i x_{i *}$ and $x^{i *}\sim_j x_{i *}$ for all $j\in N\setminus\{ i\}$. 
It is known that c-diversity is equivalent to the condition that individuals' tastes $u_1, u_2, \ldots, u_n$ are linearly independent \citep{weymark1993harsanyi}. 
Note that the latter condition is stronger than the former. 

%%%%%%%%%%%%%%%%%%%%%%%%%%%%%%%%%%%%
\section{Preference Aggregation}
\label{sec:main}
%%%%%%%%%%%%%%%%%%%%%%%%%%%%%%%%%%%%

\subsection{Pareto{\pstar}}

The following condition is a standard axiom for aggregation rules. This requires that if all individuals weakly prefer one act to another, then society should do so. 

\begin{axiom}[Pareto]
    For all $f,g\in\mathcal{F}$, if $f \succsim_i g$ for all $i\in N$, then $f \succsim_0 g$. 
\end{axiom}

We introduce the ``dual'' of the above. The following axiom requires that if no individual thinks one act $f$ to be weakly better than another $g$, then society should conclude that $f$ is not weakly better than $g$. 

\begin{axiom}[Pareto{\pstar}]
    For all $f,g\in\mathcal{F}$, if $f \not\succsim_i g$ for all $i\in N$, then $f \not\succsim_0 g$. 
\end{axiom}

Note that both principles just require society to respect unanimity.
In the standard setup where all agents have SEU preference relations, since Pareto{\pstar} is equivalent to the Paretian condition for strict preference relations, these two principles lead to almost the same consequences. 
On the other hand, when agents may have incomplete preference relations, an important difference arises. 
In such a situation, they may refrain from expressing their suggestions due to a lack of confidence. 
Although Pareto ignores agreements including unexpressed opinions, society should not ignore their silence. 
The relation $f\not\succsim_i g$ suggests that the agent has some reason preventing them from concluding that $f\succsim_i g$, which is also meaningful information for making collective decisions. 
Pareto{\pstar} is a minimal requirement ensuring that the social preference relation respects these unexpressed opinions as well.

Theorem 1 of \citet{danan2016robust} shows that under c-diversity, Pareto is equivalent to the condition that there exists $\qty(\alpha, \beta) \in \qty( \mathbb{R}_+^n \setminus \{ \mathbf{0} \} ) \times \mathbb{R}$ such that 
\begin{equation}
    u_0 = \sum_{i\in N} \alpha_i u_i + \beta 
    ~~~\text{and} ~~~
    P_0\subseteq \bigcap_{i\in N \,:\, \alpha_i > 0} P_i. 
\end{equation} 
That is, under Pareto, tastes must be aggregated via a weighted utilitarian rule, and the social belief set must be included in the intersection of belief sets of the individuals with positive weights. 
If the intersection of all individual belief sets is empty, then society has to ignore some individuals by assigning weight $0$ to them, which is unacceptable from a normative perspective.
However, the negative consequence does not occur if the individuals have some common prior. 
It has been known in the literature that when all agents have SEU preference relations, Pareto implies that society can assign positive weights to all the individuals only when they have a common belief.
In contrast, the above theorem suggests that individuals do not need to have exactly the same prior set: 
Society can assign positive weights to all individuals as long as there exists at least one probability distribution plausible for all of them. 

We then characterize the implication of the dual of Pareto---Pareto{\pstar}---under c-diversity. 
The following result shows that a severe impossibility arises in the taste aggregation even when individuals have a common set of priors. 

\begin{theorem}
\label{thm:stand}
    Suppose that for each $i\in N\cup \{ 0 \}$, $\succsim_i$ is a Bewley preference relation associated with $\qty(u_i, P_i)$ and $\qty(\succsim_i)_{i\in N}$ satisfies c-diversity. Then the following two statements are equivalent: 
     \begin{enumerate}
        \item[(i)] Pareto{\pstar} holds; 
        \item[(ii)] there exists $\qty(\alpha, \beta) \in \qty( \mathbb{R}_+^n \setminus \{ \mathbf{0} \} ) \times \mathbb{R}$ such that 
        \begin{equation}
            u_0 = \sum_{i\in N} \alpha_i u_i + \beta,
        \end{equation}
        and either 
        (a) 
        there exist a nonempty set $M\subseteq N$ and $p\in \Delta(S)$ such that $P_i = \{p\}$ for all $i\in M$, $\alpha_j = 0$ for all $j\in N\setminus M$, and $p\in P_0$,
        or 
        (b) there exists $i^* \in N$ such that $\alpha_{j} = 0$ for all $j\in N\setminus \{i^* \}$ and $P_0\supseteq P_{i^*}$. 
    \end{enumerate}
\end{theorem}

The first part of statement (ii) states that tastes should be aggregated by a weighted utilitarian rule, as in \citeauthor{harsanyi1955cardinal}'s (\citeyear{harsanyi1955cardinal}) aggregation theorem. 
The more important implication lies in the latter part. 
We classify the cases according to whether society respects only individuals with a single prior (case (a)) or not (case (b)). 
In case (a), society can assign positive weights to only individuals who have some single prior $p$. 
Therefore, it cannot respect all individuals' tastes unless they have a common single prior. 
This is similar to the impossibility in the preference aggregation of SEU individuals (e.g., \citealp{hylland1979impossibility}; \citealp{mongin1995consistent}). 
The only difference from the results in the SEU-agent model is that society can care about the beliefs of individuals who are ignored in the taste aggregation because any social belief set is allowed as long as it includes the prior of the individuals with positive weights. 
On the other hand, case (b) covers the situations where society respects individuals with non-unique priors. In this case, \textbf{there must exist a dictator $i^*$ in the taste aggregation, even if society wants to respect individuals with common multiple priors}. As in the previous case, society can take an arbitrary belief set that is large enough. Thus, this theorem suggests that respecting all individuals' tastes and beliefs is much more difficult compared to Theorem 1 of \citet{danan2016robust}---it is impossible to escape from ignoring some individuals in the taste aggregation, even if all individuals have a common belief set.

When we combine the above theorem with Theorem 1 of \citet{danan2016robust}, the following corollary can be obtained. 

\begin{corollary}
    Suppose that for each $i\in N\cup \{ 0 \}$, $\succsim_i$ is a Bewley preference relation associated with $\qty(u_i, P_i)$ and $\qty(\succsim_i)_{i\in N}$ satisfies c-diversity. Then the following two statements are equivalent: 
    \begin{enumerate}
        \item[(i)] Pareto and Pareto{\pstar} hold; 
        \item[(ii)] there exists $(\alpha, \beta) \in  \qty( \mathbb{R}_+^n \setminus \{ \mathbf{0} \} ) \times \mathbb{R}$ such that 
        \begin{equation}
            u_0 = \sum_{i\in N} \alpha_i u_i + \beta,
        \end{equation}
        and either 
        (a) there exist a nonempty set $M\subseteq N$ and $p\in \Delta(S)$ such that $P_i = \{p\}$ for all $i\in M$, $\alpha_j = 0$ for all $j\in N\setminus M$, and $P_0 =\{p\}$ or, 
        (b) there exists $i^* \in N$ such that $\alpha_{j} = 0$ for all $j\in N\setminus \{i^* \}$ and $P_0 =  P_{i^*}$. 
    \end{enumerate}
\end{corollary}

This corollary states that under Pareto and Pareto{\pstar}, society can respect all individuals' preferences only if they have a common single prior. 
If they have a common but multiple priors, then there must be a dictator $i^\ast$ in both tastes and beliefs (case (b)). 
Thus, imposing the two Paretian conditions results in a more stringent impossibility than solely requiring Pareto or Pareto{\pstar}.  

We then explain the intuition behind Theorem \ref{thm:stand}. 
The utilitarian aggregation part directly follows from \citeauthor{harsanyi1955cardinal}'s (\citeyear{harsanyi1955cardinal}) aggregation theorem (cf.~\citealp{de1995note}). 
It should be noted that Pareto{\pstar} is equivalent to the condition that for all $f,g\in \mathcal{F}$, $f\succsim_0 g$ implies $f\succsim_i g$ for some $i\in N$. 
Since we can construct a pair of acts relevant only for an arbitrary individual by c-diversity, we can show that for each individual $i$ with a positive weight, society should be more indecisive than individual $i$. 
This implies that the social belief set $P_0$ is larger than the individual set $P_i$ (see Lemma \ref{lem:large}). 
This is why the social belief set must be large enough under Pareto{\pstar}, while Pareto has the opposite effect as shown in \citet{danan2016robust}. 

The most important observation is that in almost all preference profiles, society cannot assign positive weights to more than one individual. 
This part can be shown by generalizing the argument in Section \ref{sec:example}, which shows that the equal-weight averaging rule violates Pareto{\pstar}.
That is, respecting agreements achieved by using different tastes and priors leads to a negative consequence. 
Since our proof of Theorem \ref{thm:stand} explicitly constructs a contradictory pair of acts under which spurious unanimity holds,
this illuminates the logical connection between spurious unanimity and the difficulty in respecting all individuals.

%%%%%%%%%%%%%%%%%%
\subsection{Common-Taste Pareto{\pstar}}
\label{sub:common}

We introduce two conditions that can avoid the problem of respecting spurious unanimity. 
\citet{danan2016robust} examined a weak Pareto principle that focuses on pairs of acts where all individuals have common tastes on relevant outcomes.
For a set $A$, let $\mathrm{conv}\,A$ denote the convex hull of $A$. 
We say that two acts $f$ and $g$ \textit{have no taste disagreement} if for all $x,y\in \mathrm{conv}\, \qty(f(S)\cup g(S))$ and all $i,j\in N$, $x\succsim_i y$ implies $x\succsim_j y$. 
The following restricts Pareto to these pairs of acts. 

\begin{axiom}[Common-Taste Pareto]
    For all $f,g\in\mathcal{F}$ that have no taste disagreement, if $f \succsim_i g$ for all $i\in N$, then $f \succsim_0 g$. 
\end{axiom}

Again, we consider the ``dual'' of the above axiom. 
The next requires that for all two $f$ and $g$ that have no taste disagreement, if every individual does not weakly prefer $f$ to $g$, then society should not do so. 

\begin{axiom}[Common-Taste Pareto{\pstar}]
    For all $f,g\in\mathcal{F}$ that have no taste disagreement, if $f \not\succsim_i g$ for all $i\in N$, then $f \not\succsim_0 g$. 
\end{axiom}

Theorem 2 of \citet{danan2016robust} shows that under c-minimal agreement, Common-Taste Pareto holds if and only if there exists $\qty(\alpha, \beta) \in \qty( \mathbb{R}_+^n \setminus \{ \mathbf{0} \} ) \times \mathbb{R}$ such that 
\begin{equation}
\label{eq:DGHT_thm2}
    u_0 = \sum_{i\in N} \alpha_i u_i + \beta ~~~\text{and} ~~~P_0 \subseteq \mathrm{conv}\,\qty( \bigcup_{i\in N} P_i  ). 
\end{equation}
Note that the latter condition of \eqref{eq:DGHT_thm2} can be rewritten as for each $p_0 \in P_0$, there exists $(p_1,p_2, \ldots, p_n)\in \Pi_{i\in N} P_i$ such that for some $\gamma \in \Delta(N)$, 
\begin{equation}
    p_0 = \sum_{i\in N} \gamma_i p_i. 
\end{equation}
Thus, according to their theorem, tastes are aggregated in a utilitarian way, and each social belief can be written 
as a weighted sum of some combination of beliefs plausible for each individual.%
\footnote{
\citet{danan2013aggregating,danan2015harsanyi} studied aggregation of risk preferences in the situation where individuals and society have expected multi-utility preferences \citep{dubra2004expected}. 
These preferences are parametrized by a set of utility functions, and compare lotteries as in the Bewley model. 
\citet{danan2015harsanyi} showed that the Pareto principle has a similar implication in their setup. 
}

As in the previous theorem, we examine the implications of Common-Taste Pareto{\pstar}. 
The next theorem shows that this axiom can be characterized by a new class of belief aggregation rules, which can be seen as the dual of \eqref{eq:DGHT_thm2}.  

\begin{theorem}
\label{thm:common}
    Suppose that for each $i\in N\cup \{ 0 \}$, $\succsim_i$ is a Bewley preference relation associated with $\qty(u_i, P_i)$ and $\qty(\succsim_i)_{i\in N}$ satisfies c-minimal agreement. Then the following two statements are equivalent: 
    \begin{enumerate}
        \item[(i)] Common-Taste Pareto{\pstar} holds; 
        \item[(ii)] there exists $\qty(\alpha, \beta) \in \qty(\mathbb{R}_+^n \setminus \{ \mathbf{0} \}) \times \mathbb{R}$ such that 
        \begin{equation}
            u_0 = \sum_{i\in N} \alpha_i u_i + \beta,
        \end{equation}
        and for all $(p_1,  p_2,\ldots, p_n) \in \Pi_{i\in N} P_i$, there exists $p_0 \in P_0$ such that for some $\gamma\in \Delta(N)$, 
        \begin{equation}
            p_0 = \sum_{i\in N} \gamma_i p_i. 
        \end{equation}
    \end{enumerate}
\end{theorem}

Our aggregation rule states that each combination $(p_1,  p_2,\ldots, p_n)$ of priors in $\Pi_{i\in N} P_i$ must be used in the social evaluation.
Thus, Common-Taste Pareto{\pstar} rules out disposing of any combination of plausible priors. 
This implication is complementary to that of Common-Taste Pareto: 
While Common-Taste Pareto requires that all social beliefs should be constructed from individual ones, 
Common-Taste Pareto{\pstar} claims that \textbf{all combinations of individual beliefs must constitute some social belief.} 
In other words, while the corresponding theorem in \citet{danan2016robust} determines an ``upper bound'' of the social belief set, our theorem offers a ``lower bound'' of the social belief set because it must be large enough under the aggregation rules in Theorem \ref{thm:common}. 
Combined with the theorem of \citet{danan2016robust}, our result suggests further normatively appealing ways of constructing a social belief set within the class of aggregation rules characterized in \citet{danan2016robust}. 

Note that several papers, such as \citet{cres2011aggregation}, \citet{hill2012unanimity}, \citet{alon2016utilitarian}, \citet{hayashi2023collective}, and \citet{dong2024aggregation}, also studied preference aggregation within frameworks where agents have imprecise beliefs. To avoid respecting spurious unanimity, these papers imposed weak Paretian requirements or focused only on cases without taste disagreements. Although some of the aggregation rules characterized in these papers are included in our class of rules, the same rule has not been obtained in the literature. Compared with these papers, Theorem \ref{thm:common} provides a foundation for a broader class of more fundamental aggregation rules.

One important observation from Theorem \ref{thm:common} is that under Common-Taste Pareto{\pstar}, the probability distribution 
plausible for all the individuals must be in the social prior set. To see this, let $p\in \bigcap_{i\in N}P_i$. Then,
by Theorem \ref{thm:common}, there exist $p_0 \in P_0$ and weight $\gamma\in \Delta(N)$ such that $p_0 = \sum_{i\in N}\gamma_i p = p$, which implies $\bigcap_{i\in N}P_i \subseteq P_0$. 
Furthermore, using Theorem 2 of \citet{danan2016robust}, we can see that under the two common-taste Paretian axioms, if all individuals have a common belief set $P$, then the social belief set should coincide with $P$. 
These natural properties cannot be ensured only by Common-Taste Pareto. 
That is, Common-Taste Pareto{\pstar} captures the important but overlooked aspect of collective decisions. 
We provide these implications as formal statements. 

\begin{corollary}
    Suppose that for each $i\in N\cup \{ 0 \}$, $\succsim_i$ is a Bewley preference relation associated with $\qty(u_i, P_i)$ and $\qty(\succsim_i)_{i\in N}$ satisfies c-minimal agreement. Then the following two statements hold:
    \begin{enumerate}
        \item[(a)] Common-Taste Pareto{\pstar} implies $\bigcap_{i\in N} P_i \subseteq P_0$; 
        \item[(b)] Common-Taste Pareto and Common-Taste Pareto{\pstar} imply that if there is $P \subseteq \Delta(S)$ such that $P_i = P$ for all $i\in N$, then $P_0 =P$. 
    \end{enumerate}
\end{corollary}

Finally, we make a remark on the compatibility of aggregation rules obtained in the above theorem and Theorem 2 of \citet{danan2016robust}. 
As explained, Common-Taste Pareto and Common-Taste Pareto{\pstar} have the opposite effects, but we can see that there always exist social beliefs compatible with both implications. 
For instance, the sets $P_0$ such that $P_0 = \sum_{i\in N}\gamma_i P_i$ for some $\gamma\in \Delta(N)$ or $P_0 =\mathrm{conv}\,\big(\bigcup_{i\in N} P_0\big)$ satisfy the two conditions. 
Thus, we can always impose both common-taste Paretian axioms at the same time. 

%%%%%%%%%%%%%%%%%%%%%%%%%%%%%%%%%%%%
\section{Discussion}
\label{sec:dis}
%%%%%%%%%%%%%%%%%%%%%%%%%%%%%%%%%%%%

\subsection{Belief Exchange}

In Section \ref{sec:main}.\ref{sub:common}, we have provided a possibility result by focusing on the Paretian principles restricted to pairs of acts that have no taste disagreement. 
Since these conditions require the social preference to respect unanimity in very specific situations, 
one might be concerned that the obtained aggregation rules may ignore other non-spurious agreements: for instance, agreements among individuals with minor differences in tastes. 
Here, we argue that the aggregation rules studied in Section \ref{sec:main}.\ref{sub:common} can respect reasonable agreement that are not covered by the common-taste Paretian principles. 

Consider a hypothetical situation in which individuals suspect that their beliefs might be incorrect, and hence, compare acts by using not only their own prior sets but also those of others---that is, by imaginarily exchanging beliefs. If they have opposite tastes and beliefs, then they cannot reach agreement under this thought experiment. By contrast, when the difference in their preferences (especially their tastes) is not too large, they can attain agreement even in this hypothetical scenario. \citet{billot2021utilitarian} introduced a restricted Paretian principle that focuses on these robust consensuses in the SEU-agent model. 
We extend their axiom to our model. 

\begin{axiom}[Exchange Pareto]
    For all $f,g\in \mathcal{F}$, if $\sum_{s\in S}p_j(s)u_i(f(s))\geq \sum_{s\in S}p_j(s)u_i(g(s))$ for all $i,j\in N$ and all $p_j\in P_j$, then $f\succsim_0 g$. 
\end{axiom}

We then define the ``dual'' requirement.
It focuses on hypothetical situations where every belief set contains a prior that prevents each individual from judging $f$ to be at least as good as $g$.
The axiom below requires that when all individuals are in such a position, the social preference should likewise refrain from deeming $f$ weakly better than $g$.

\begin{axiom}[Exchange Pareto{\pstar}]
    For all $f,g\in \mathcal{F}$, if there exists $(p_1,\ldots, p_n) \in \prod_{j\in N} P_j$ such that  $\sum_{s\in S} p_j(s) u_i(g(s)) > \sum_{s\in S} p_j (s) u_i(f(s))$ for all $i,j\in N$, then $f \not\succsim_0 g$. 
\end{axiom}

Note that, unlike Common-Taste Pareto and Common-Taste Pareto{\pstar}, the above two axioms do not directly restrict the pairs of acts---we can apply them in the cases where individuals have different tastes on related outcomes, as long as the conditions in the axioms are satisfied. 
The following propositions show that each of these two principles has the same implications as its corresponding common-taste Paretian axiom. 

\begin{proposition}
\label{prop:ex}
    Suppose that for each $i\in N\cup \{ 0 \}$, $\succsim_i$ is a Bewley preference relation associated with $\qty(u_i, P_i)$ and $\qty(\succsim_i)_{i\in N}$ satisfies c-minimal agreement. Then the following two statements are equivalent: 
    \begin{enumerate}
        \item[(i)] Exchange Pareto holds; 
        \item[(ii)] there exists $\qty(\alpha, \beta) \in \qty( \mathbb{R}_+^n \setminus \{ \mathbf{0} \} ) \times \mathbb{R}$ such that $u_0 = \sum_{i\in N} \alpha_i u_i + \beta$ and $P_0 \subseteq \textup{conv}\, \qty( \bigcup_{i\in N} P_i )$. 
    \end{enumerate}
\end{proposition}

\begin{proposition}
\label{prop:exstar}
    Suppose that for each $i\in N\cup \{ 0 \}$, $\succsim_i$ is a Bewley preference relation $\qty(u_i, P_i)$ and $\qty(\succsim_i)_{i\in N}$ satisfies c-minimal agreement. Then the following two statements are equivalent: 
    \begin{enumerate}
        \item[(i)] Exchange Pareto{\pstar} holds; 
        \item[(ii)] there exists $\qty(\alpha, \beta) \in \qty( \mathbb{R}_+^n \setminus \{ \mathbf{0} \} ) \times \mathbb{R}$ such that $ u_0 = \sum_{i\in N} \alpha_i u_i + \beta$, 
        and for all $(p_1,  p_2,\ldots, p_n) \in \Pi_{i\in N} P_i$, there exists $p_0 \in P_0$ such that $p_0 = \sum_{i\in N} \gamma_i p_i$ for some $\gamma\in \Delta(N)$. 
    \end{enumerate}
\end{proposition}

It is easy to see that Exchange Pareto (resp.~Exchange Pareto{\pstar}) is stronger than Common-Taste Pareto (resp.~Common-Taste Pareto{\pstar}). 
The above results suggest that aggregation rules considered in Section \ref{sec:main}.\ref{sub:common} can respect a broader range of agreements than those taken into account by Common-Taste Pareto or Common-Taste Pareto{\pstar}:
Small divergences in tastes are permissible as long as they do not lead to the cancellation of double disagreements.

%%%%%%%%%%%%%%%%%%
\subsection{Common-Taste Pareto{\pstar} and complete social preference relations}

We have studied social preference relations that may violate completeness. 
However, society is sometimes required to have a complete preference relation in order to determine which alternatives are socially desirable in any situation.
In this part, we discuss the relationship between our second theorem and the completeness of the social preference relation. 

One approach to deal with this issue is to assume that $\succsim_0$ is an SEU preference relation. 
However, under Common-Taste Pareto{\pstar}, there may be no compatible SEU social preference relation because this axiom requires the social belief set to be large enough. 
For instance, suppose that there is a pair $(f,g)$ of acts such that all individuals cannot compare them (i.e., $f\not \succsim_i g$ and $g\not \succsim_i f$ for all $i\in N$).
Then, under completeness, Common-Taste Pareto{\pstar} implies $g\succ_0 f$ and $f\succ_0 g$, but they are incompatible. 
From this observation, the following question arises: When can the SEU social preference relation satisfy Common-Taste Pareto{\pstar}?

For each profile $(P_1, \ldots, P_n)$ of belief sets, let $\mathbb{P}(P_1, \ldots, P_n)$ be the collection of subsets of $\Delta(S)$ such that
\begin{equation*} 
    P\in \mathbb{P}(P_1, \ldots, P_n)
    \iff 
    \qty[~ P = \mathrm{conv}\, \{ p_1,\ldots, p_n \} ~~\text{for some} ~~ (p_1,\ldots, p_n) \in \Pi_{i\in N}P_i  ~]. 
\end{equation*}
The following result characterizes the condition that the SEU social preference relation can satisfy Common-Taste Pareto{\pstar}. 

\begin{corollary}
    Suppose that for each $i\in N$, $\succsim_i$ is a Bewley preference relation $\qty(u_i, P_i)$, $\succsim_0$ is an SEU preference relation $(u_0, p_0)$, and $\qty(\succsim_i)_{i\in N}$ satisfies c-minimal agreement. Then the following two statements are equivalent: 
    \begin{enumerate}
        \item[(i)] there exists $(u_0, p_0)$ such that Common-Taste Pareto{\pstar} holds; 
        \item[(ii)] $\bigcap_{P\in \mathbb{P}(P_1, \ldots, P_n)} P \neq \emptyset$. 
    \end{enumerate}
\end{corollary}

Statement (ii) characterizes the condition for the existence of compatible SEU preference relations by the property of parameters. 
Consider a hypothetical scenario in which each individual chooses a probability distribution from their own belief sets, and society forms its belief set by taking the convex hull of all the chosen distributions.
Then, the second statement claims that under this procedure, some probability distribution is always included in the hypothetical belief set constructed as above, regardless of individuals' choices.
To see why the above equivalence holds, note that the existence of a probability distribution $p\in \bigcap_{P\in \mathbb{P}(P_1, \ldots, P_n)} P$ means that for each combination $(p_1,\ldots, p_n)\in \prod_{i\in N}P_i $, there exists $\gamma\in \Delta(N)$ such that $p = \sum_{i\in N}\gamma_i p_i$.
Then, by Theorem \ref{thm:common}, we can set $p_0 = p$ without violating Common-Taste Pareto{\pstar}. 

Another approach is to treat incomplete multiprior preference relations as ``objective rationality'' relations \citep{gilboa2010objective} and to separate the problem of preference aggregation from that of constructing complete rankings. 
For the latter part, a general completion model was considered in \citet{danan2016robust}. 
They proposed utility functions that evaluate each act $f\in \mathcal{F}$ by an act-dependent weighted sum of the minimum expected utility---$\min_{p\in P_0} \sum_{s\in S} p(s) u_0(f(s))$---and the maximum one---$\max_{p\in P_0} \sum_{s\in S} p(s) u_0(f(s))$. 
For the cautious completion procedures, \citet{gilboa2010objective} characterized the maxmin expected utility model based on Bewley preference relations. 
In this framework, foundations for the $\varepsilon$-contamination model \citep{kopylov2009choice}, the uncertainty-averse model, the variational model  \citep{cerreia2016objective}, and the $\alpha$-maxmin expected utility model \citep{frick2022objective} were also provided.
It should be noted that social preference relations constructed in this way do not satisfy Common-Taste Pareto{\pstar} itself; however, they are compatible with this axiom in a weak sense, as discussed in the following subsection. 
This approach allows us to construct a complete social ranking under arbitrary preference profiles. 

%%%%%%%%%%%%%%%%%%
\subsection{Application: aggregating MBA preference relations}

The scope of our theorems, especially Theorem \ref{thm:common}, is not only on incomplete preference aggregation. 
As considered in \citet{danan2016robust} and \citet{Mudekereza2025WP}, our results are applicable to the aggregation of utility functions of ambiguity preferences with regularity properties. 
Suppose that each $\succsim_i$ is a monotone Bernoullian
Archimedean (MBA) preference relation; that is, $\succsim_i$ satisfies nontriviality, completeness, transitivity, continuity, statewise monotonicity, and risk independence \citep{cerreia2011rational}. 
Then, $\succsim_i$ is represented by a function $U_i$ such that for all $f\in \mathcal{F}$, 
\begin{equation}
    U_i (f) =  \alpha_i (f) \min_{p\in P_i} \sum_{s\in S} p(s) u_i(f(s)) + \qty(1 - \alpha_i (f)) \max_{p\in P_i}  \sum_{s\in S} p(s) u_i(f(s)), 
\end{equation}
where $u_i$ is a real-valued nonconstant affine function, $P_i$ is a nonempty compact convex subset of $\Delta(S)$, and $\alpha_i$ is a function from $\mathcal{F}$ to $[0,1]$. 
It is known that for $\succsim_i$, the Bewley preference relation $\qty(u_i, P_i)$ is the maximal subrelation $\succsim^C_i$ that satisfies independence; that is, for all $f,g\in \mathcal{F}$, 
\begin{equation}
    f \succsim^C_i g \iff \qty[~ \alpha f + (1-\alpha ) h \succsim_i  \alpha g + (1-\alpha) h  ~~~\text{for all $h\in \mathcal{F}$ and all $\alpha\in [0,1]$}~]. 
\end{equation}

We say that \textit{Revealed Common-Taste Pareto{\pstar}} holds for a profile $\qty(\succsim_i)_{i\in N\cup \{ 0 \}}$ of MBA preference relations if for all $f,g\in\mathcal{F}$, 
\begin{equation*}
    \qty[~ f\not\succsim_i^C g ~~~\text{for all $i\in N$} ~] \implies f\not\succsim_0^C g. 
\end{equation*}
By the construction of $\succsim^C_i$, this property is equivalent to the condition that for all $f,g \in \mathcal{F}$, if there exists $\qty(\beta, (h_1, \ldots, h_n)) \in [0,1]^n \times \mathcal{F}^n$ such that $\beta_i g +(1-\beta_i ) h_i \succ_i \beta_i f +(1-\beta_i ) h_i$ for each $i\in N$, then 
 $\beta_0 g +(1-\beta_0 ) h_0 \succ_0 \beta_0 f +(1-\beta_0 ) h_0$ must hold for some $(\beta_0,h_0 ) \in [0,1] \times  \mathcal{F}$. 
Thus, Revealed Common-Taste Pareto{\pstar} requires that if $f$ is not robustly better than $g$ in the above sense for all individuals, then society should not have the robust relation either. 
 
Then, Theorem 2 is applicable for aggregating MBA preferences and provides a guideline on how the social baseline belief set $P_0$ should be formed. 
For the aggregation of maxmin expected utility preference relations \citep{gilboa1989maxmin}, our result offers a direct restriction on the social belief set. 
As another special case, when aggregating variational preference relations \citep{maccheroni2006ambiguity}, our aggregation rule determines which belief should be taken into account since the effective domain of the perception function coincides with the belief set characterizing the subrelation $\succsim^C_0$ \citep{CMMR2015}. 

%%%%%%%%%%%%%%%%%%%%%%%%%%%%%%%%%%%%
\section*{APPENDIX}
%%%%%%%%%%%%%%%%%%%%%%%%%%%%%%%%%%%%

\renewcommand{\thesubsection}{A.\arabic{subsection}}
\setcounter{subsection}{0}

For each $s\in S$, the degenerated probability distribution on state $s$ is denoted by $\delta_s$. 
For a subset $A$ and an element $a$ of $X$ or $\mathbb{R}$, let $A + a$ denote the set $\{ a' + a \mid a'\in A\}$. 

\subsection{Proof of Theorem \ref{thm:stand}}

We only prove that (i) implies (ii). Note that since each $\succsim_i$ is complete on $X$, Pareto{\pstar} implies that for all $x,y\in X$, if $x\succ_i y$ for each $i\in N$, then $x\succ_0 y$. 
Since each $u_i$ is an affine function, we can apply Proposition 2 of \citet{de1995note}. 
Thus, there exists $(\alpha, \beta) \in  \qty( \mathbb{R}_+^n \setminus \{ \mathbf{0} \} ) \times \mathbb{R}$ such that $u_0 = \sum_{i\in N} \alpha_i u_i + \beta$.

We then show the latter part of statement (ii). 
Since $\qty(\succsim_i)_{i\in N}$ satisfies c-diversity, for each $i\in N$, there exist $x^{i*}, x_{i*} \in X$ such that $x^{i*} \succ_i x_{i*}$ and $x^{i*} \sim_j  x_{i*}$ for all $j\in N\setminus \{ i \}$. 
The following intermediate result holds. 

\begin{lemma}
\label{lem:large}
    Pareto{\pstar} implies that $P_0 \supseteq \bigcup_{i\in N \,:\, \alpha_i > 0} P_i$. 
\end{lemma}

\begin{proof}
    Let $i\in N$ be such that $\alpha_i > 0$. Suppose to the contrary that $P_i$ is not a subset of $P_0 $. 
    Then, there exists $p_i\in P_i\setminus P_0$. 
    By the separating hyperplane theorem, there exists  $(\lambda, \kappa) \in \qty( \mathbb{R}^S\setminus \{ \mathbf{0} \} ) \times \mathbb{R}$ such that for all $p_0\in P_0$, 
    \begin{equation}
    \label{eq:hyperplane_P*}
    \sum_{s\in S} p_i(s) \lambda (s)  > \kappa > \sum_{s\in S} p_0(s) \lambda (s). 
    \end{equation}
    Without loss of generality, we can assume that for each $s\in S$, 
    \begin{equation}
    \label{eq:lem_normalize}
        \lambda (s), \kappa \in \qty[ \sum_{j\in N} {1\over n} u_i( x^{j*}), \sum_{j\in N} {1\over n} u_i (x_{j*})] . 
    \end{equation} 
    Take $\varepsilon \in (0, {1\over2}]$ arbitrarily.
    Let $f^\varepsilon\in \mathcal{F}$ be such that for each $s\in S$, 
    \begin{equation}
        f^\varepsilon(s)\in {1\over n}\mathrm{conv}\, \{ x^{i*}, x_{i*} \} + \sum_{j\in N\setminus \{i \}} {1\over n} \qty( \qty({1\over 2} + \varepsilon) x^{j*} +  \qty({1\over 2} - \varepsilon) x_{j*})
    \end{equation}
     and $u_i (f^\varepsilon (s) ) = \lambda(s)$.
     Furthermore, let $x^\varepsilon \in X$ be such that 
     \begin{equation}
         x^\varepsilon \in  {1\over n}\mathrm{conv}\, \{ x^{i*}, x_{i*} \} + \sum_{j\in N\setminus \{i \}} {1\over n} \qty( \qty({1\over 2} - \varepsilon) x^{j*} +  \qty({1\over 2} + \varepsilon) x_{j*})
    \end{equation}
     and $u_i (x^\varepsilon) = \kappa$. 
     Note that we can take such $f^\varepsilon$ and $x^\varepsilon$ by \eqref{eq:lem_normalize}. 
    Then, the first inequality in \eqref{eq:hyperplane_P*} can be rewritten as $\sum_{s\in S} p_i(s) u_i (f^\varepsilon(s)) >  u_i (x^\varepsilon)$, which implies that $x^\varepsilon \not\succsim_i f^\varepsilon$. 
    For each $j\in N\setminus \{i\}$, by the constructions of $x^\varepsilon$ and $f^\varepsilon$, we have that for each $s\in S$, 
    \begin{equation}
        u_j(f^\varepsilon (s)) - u_j(x^\varepsilon ) = {2\varepsilon\over n} \qty(u_j (x^{j*}) - u_j (x_{j*})) > 0, 
    \end{equation}
    which implies $x^\varepsilon \not\succsim_j f^\varepsilon$. 
    By Pareto{\pstar}, $x^\varepsilon \not\succsim_0 f^\varepsilon $ must hold. 

    By the second inequality in \eqref{eq:hyperplane_P*}, $ \alpha_i u_i (x^\varepsilon) - \sum_{s\in S} p_0(s) \alpha_i u_i (f^\varepsilon (s)) >0 $  for all $p_0\in P_0$ and the left-hand side is constant with respect to $\varepsilon$. 
    For each $j\in N\setminus\{i\}$ and each $p_0\in P_0$,  we have 
    \begin{equation}
        \sum_{s\in S} p_0(s) \alpha_j u_j (f^\varepsilon (s)) - \alpha_j u_j (x^\varepsilon)  = \alpha_j {2\varepsilon\over n} \qty(u_j (x^{j*}) - u_j (x_{j*})) > 0. 
    \end{equation}
    Thus, we can set $\varepsilon$ small enough that for all $p_0\in P_0$,
    \begin{align*}
        \sum_{j\in N\setminus \{i \}} \sum_{s\in S} p_0(s) \alpha_j u_j (f^\varepsilon (s)) -  \sum_{j\in N\setminus \{i \}} \alpha_j u_j (x^\varepsilon)  
        &< \alpha_i u_i (x^\varepsilon) - \sum_{s\in S} p_0(s) \alpha_i u_i (f^\varepsilon (s)). 
    \end{align*}
    Therefore, by $u_0 = \sum_{j\in N} \alpha_j u_j + \beta$, we have  $ u_0 (x^\varepsilon)  > \sum_{s\in S} p_0(s) u_0 (f^\varepsilon (s))$ for all $p_0\in P_0$, which means that $x^\varepsilon \succsim_0 f^\varepsilon$. This is a contradiction to the result of the last paragraph. 
\end{proof}

Let $N^+ = \{ i\in N \mid \alpha_i > 0\}$. 
First, we show that if there exists $i\in N^+$ such that $P_i$ is not a singleton, then $N^+$ is a singleton. 
Suppose otherwise. 
Without loss of generality, 
we assume that $\alpha_1, \alpha_2 >0$ and $P_1$ or $P_2$ is not a singleton. 
Then, we can take distinct $p_1 \in P_1$ and $p_2 \in P_2$. 
Define the real-valued functions $\varphi_{1}, \varphi_{2}$ on $\Delta(S)$ as for all $p\in \Delta(S)$, 
\begin{align*}
    \varphi_{1} (p)
    &= 
    \alpha_2 \sum_{s\in S} (p_1(s) - p_2(s)) \qty{ p(s) - {p_1(s) + p_2(s)\over 2}}  -\varepsilon ,
    \\
    \varphi_{2} (p)
    &= 
    \alpha_1 \sum_{s\in S} (p_2(s) - p_1(s)) \qty{ p(s) - {p_1(s) + p_2(s)\over 2}}  -\varepsilon, 
\end{align*}
where $\varepsilon \in \mathbb{R}_{++}$ is small enough that $\varepsilon < {1 \over 2} \min \{ \alpha_1, \alpha_2 \} \sum_{s\in S} (p_2 (s) - p_1(s))^2$. 
Note that 
\begin{equation}
\label{eq:var1>half}
    \varphi_1 (p_1) = {\alpha_2\over 2}  \sum_{s\in S} \qty(p_1 (s) - p_2(s))^2 - \varepsilon > 0
\end{equation} 
and 
\begin{equation}
\label{eq:var2>half}
    \varphi_2 (p_2) = {\alpha_1\over 2}  \sum_{s\in S} \qty(p_2 (s) - p_1(s))^2  -\varepsilon > 0. 
\end{equation}

Take $\nu\in (0,{1\over 2}]$ arbitrarily. 
Let 
\begin{equation}
    \hat{x}^\nu = \sum_{i =1,2} \qty( {1\over 2n}x^{i*} + {1\over 2n}x_{i*}) + {1 \over n} \sum_{i\in N\setminus \{1,2\} } \qty( \qty({1\over 2} - \nu) x^{i*} +  \qty({1\over 2} + \nu) x_{i*}). 
\end{equation}
Since $u_i(x^{j*}) = u_i(x_{j*})$ for each $i, j\in N$ such that $i\neq j$, we can 
take $b \in \mathbb{R}_{++}$ so that for all $s\in S$, 
\begin{equation*}
    b \varphi_1 (\delta_s) + u_1 (\hat{x}^\nu)  
    \in {1\over n} \qty[ u_1 (x_{1*} )  , u_1 (x^{1*} ) ] 
    +  {1 \over n} \sum_{i\in N\setminus \{1\} } u_1( x^{i*})
\end{equation*} 
and 
\begin{equation*}
    b \varphi_2 (\delta_s) + u_2 (\hat{x}^\nu)  
    \in {1\over n} \qty[ u_2 (x_{2*} )  , u_2 (x^{2*} ) ] 
    +  {1 \over n} \sum_{i\in N\setminus \{2\} } u_2( x^{i*}).   
\end{equation*} 
Let $\hat{f}^\nu \in \mathcal{F}$ be such that for each $s\in S$, 
\begin{equation}
    \hat{f}^\nu(s) \in {1\over n}  \mathrm{conv}\,\{ x^{1*}, x_{1*}
 \} + {1\over n}\mathrm{conv}\,\{  x^{2*}, x_{2*}\}  + {1 \over n} \sum_{i\in N\setminus \{1,2\} } \qty( \qty({1\over 2} + \nu) x^{i*} +  \qty({1\over 2} - \nu) x_{i*}),
\end{equation}
$ u_1 (\hat{f}^\nu (s)) = b \varphi_1 (\delta_s) + u_1 (\hat{x}^\nu) $, and  $u_2 (\hat{f}^\nu (s)) = b \varphi_2 (\delta_s) + u_2 (\hat{x}^\nu) $. 

Then, by \eqref{eq:var1>half} and \eqref{eq:var2>half},  we have that for each $i\in\{1,2 \}$, 
\begin{equation}
\label{eq:x_uundominated_byf}
    \sum_{s\in S} p_i (s) u_i(\hat{f}^\nu(s)) = \sum_{s\in S} p_i (s) ( b \varphi_i (\delta_s) + u_i (\hat{x}^\nu) ) = b \varphi_i(p_i) +u_i (\hat{x}^\nu) > u_i (\hat{x}^\nu), 
\end{equation}
which implies that $\hat{x}^\nu \not\succsim_1 \hat{f}^\nu$ and $\hat{x}^\nu \not\succsim_2 \hat{f}^\nu$.  
For each $i\in N\setminus \{1, 2\}$, by 
\begin{equation}
    u_i (\hat{f}^\nu (s)) - u_i (\hat{x}^\nu) = {2\nu \over n} \qty( u_i(x^{i*}) -  u_i(x_{i*}) ) > 0, 
\end{equation} 
for all $s\in S$, 
we have $\hat{x}^\nu \not\succsim_i \hat{f}^\nu$. 
By Pareto{\pstar}, we have $\hat{x}^\nu \not\succsim_0 \hat{f}^\nu$. 

Note that 
\begin{align}
    \alpha_1 u_1 (\hat{f}^\nu (s)) + \alpha_2 u_2 (\hat{f}^\nu (s)) 
    &= \alpha_1 \qty( b \varphi_1 (\delta_s) + u_1 (\hat{x}^\nu) ) + \alpha_2 \qty( b\varphi_2 (\delta_s) + u_2 (\hat{x}^\nu) )
    \\
    &= \alpha_1 u_1(\hat{x}^\nu) + \alpha_2 u_2(\hat{x}^\nu) - b (\alpha_1 + \alpha_2)\varepsilon. 
\end{align}
Thus, for all $p_0\in P_0$, we have 
\begin{equation}
     \sum_{i = 1,2} \alpha_i u_i (\hat{x}^\nu) -  \sum_{s\in S} p_0(s) \sum_{i = 1,2} \alpha_i  u_i (\hat{f}^\nu (s)) = b (\alpha_1 + \alpha_2)\varepsilon > 0, 
\end{equation} 
and the left-hand side is constant with respect to $\nu$. 
Also, by the constructions of $\hat{x}^\nu$ and $\hat{f}^\nu$, for all $p_0\in P_0$
\begin{equation}
    \sum_{s\in S}p_0(s) \sum_{i\in N\setminus\{1,2\}}\alpha_i u_i (\hat{f}^\nu (s)) - \sum_{i\in N\setminus\{1,2\}} \alpha_iu_i (\hat{x}^\nu) 
    = {2\nu \over n} \sum_{i\in N\setminus\{1,2\}} \alpha_i \qty( u_i(x^{i*}) -  u_i(x_{i*}) ) >0.
\end{equation}
We can take $\nu$ small enough that for all $p_0\in P_0$, 
\begin{equation}
     \sum_{s\in S}p_0 (s) \sum_{i\in N\setminus\{1,2\}}\alpha_i u_i (\hat{f}^\nu (s)) - \sum_{i\in N\setminus\{1,2\}} \alpha_iu_i (\hat{x}^\nu) < \sum_{i = 1,2} \alpha_i u_i (\hat{x}^\nu) -  \sum_{s\in S} p_0(s) \sum_{i = 1,2} \alpha_i  u_i (\hat{f}^\nu (s)). 
\end{equation}
Therefore, by $u_0 = \sum_{i\in N} \alpha_i u_i + \beta$, we have  $ u_0 (\hat{x}^\nu)  > \sum_{s\in S} p_0(s) u_0 (\hat{f}^\nu (s))$ for all $p_0\in P_0$,
which implies $\hat{x}^\nu  \succsim_0 \hat{f}^\nu$. 
This is a contradiction to the result of the previous paragraph. 

Next, we prove that if $P_i$ is a singleton for each $i\in N^+ $, then we have $P_j = P_k$ for all $j,k\in N^+$.
Suppose otherwise. 
Again, without loss of generality, 
we assume that $\alpha_1, \alpha_2>0$ but $P_1 $ and $ P_2$ are distinct singletons. 
Let $p_1, p_2\in \Delta(S)$ be such that  $P_1 = \{ p_1 \}$ and $P_2 = \{ p_2 \}$. 
We construct the functions $\varphi_1$ and $\varphi_2$, and, for each $\nu\in (0,{1\over 2}]$, define $\hat{x}^\nu\in X$ and $\hat{f}^\nu \in \mathcal{F}$ as in the first case. 
Then, for each $i\in \{ 1,2 \}$, \eqref{eq:x_uundominated_byf} holds again, which implies that $\hat{x}^\nu \not\succsim_1 \hat{f}^\nu$ and $\hat{x}^\nu \not\succsim_2 \hat{f}^\nu$. 
By construction, $\hat{x}^\nu \not\succsim_j \hat{f}^\nu$ for all $j\in N\setminus \{1, 2\}$.
By Pareto{\pstar}, $\hat{x}^\nu \not\succsim_0 \hat{f}^\nu$ holds. 
However, by the same argument as in the first case, we can take $\nu$ small enough that $u_0(\hat{x}^\nu ) > \sum_{s\in S} p_0(s) u_0(\hat{f}^\nu (s))$ for all $p_0\in P_0$. 
Thus, we obtain $\hat{x}^\nu \succsim_0 \hat{f}^\nu$, which is a contradiction.

%%%%%%%%%%%%%%%%%%
\subsection{Proof of Theorem \ref{thm:common}}

First, we prove that (i) implies (ii). 
By c-minimal agreement, there exist $x^*, x_* \in X$ such that $x^*\succ_i  x_*$ for all $i\in N$.
Without loss of generality, we assume that $u_i(x^*) = 1$ and $u_i(x_*) = 0$ for each $i\in N$.

Since each $\succsim_i$ is complete on $X$, Common-Taste Pareto{\pstar} implies that for all $x,y\in X$, if $x\succ_i y$ for each $i\in N$, then $x\succ_0 y$. 
Since each $u_i$ is an affine function, we can apply Proposition 2 of \citet{de1995note} again. 
Thus, there exists $(\alpha, \beta) \in  \qty( \mathbb{R}_+^n \setminus \{ \mathbf{0} \} ) \times \mathbb{R}$ such that $u_0 = \sum_{i\in N} \alpha_i u_i + \beta$. 

We show that for all $(p_1,  p_2,\ldots, p_n) \in \Pi_{i\in N} P_i$, there exists $p_0 \in P_0$ such that for some $\gamma\in \Delta(N)$, 
\begin{equation*}
    p_0 = \sum_{i\in N} \gamma_i p_i. 
\end{equation*}
Suppose to the contrary that for some $(p_1,  p_2,\ldots, p_n) \in \Pi_{i\in N} P_i$, there is no $p_0 \in P_0$ such that $p_0 \in \mathrm{conv}\,  \{ p_1,  p_2,\ldots, p_n \}$. 
By the separating hyperplane theorem, there exists 
$(\lambda, \kappa) \in \qty( \mathbb{R}^S\setminus \{ \mathbf{0} \} ) \times \mathbb{R}$ such that for all $i\in N$ and all $p_0\in P_0$, 
\begin{equation}
\label{eq:hyperplane_CTP*}
    \sum_{s\in S} p_i(s) \lambda (s)  > \kappa > \sum_{s\in S} p_0(s) \lambda (s). 
\end{equation}
Without loss of generality, we can take $\lambda$ and $\kappa$ so that $\lambda(s), \kappa \in [0,1]$ for each $s\in S$. 
Let $f\in \mathcal{F}$ be such that $f(s) \in \mathrm{conv}\, \{x^*, x_*\}$ and $u_i (f(s)) = \lambda(s)$ for all $s\in S$ and all $i\in N\cup\{0\}$, and let $x\in X$ be such that $x \in \mathrm{conv}\, \{x^*, x_*\}$ and $u_i (x) =\kappa $ for  all $i\in N\cup\{0\}$.
By construction, $f$ and $x$ have no taste disagreement. 
% Note that we can construct such an act and an outcome by taking $f(s)$ for each $s\in S$ and $x$ from $\mathrm{conv}\, \{x^*, x_*\}$. 
Then, \eqref{eq:hyperplane_CTP*} and $u_0 = \sum_{i\in N} \alpha_i u_i + \beta$ imply that $ \sum_{s\in S} p_i(s) u_i(f (s))   > u_i(x)$ 
for all $i\in N$ and $u_0 (x)> \sum_{s\in S} p_0(s) u_0(f(s))$ all $p_0\in P_0$. Therefore, $x\not\succsim_i f$ for all $i\in N$ but $x\succsim_0 f$, which is a contradiction to Common-Taste Pareto{\pstar}. 

We then prove the converse. Let $f, g\in\mathcal{F}$ be  such that they have no taste disagreement and $f\not\succsim_i g$ for all $i\in N$. 
By construction, each $u_i$ represents the same binary relation on $ \mathrm{conv}\, \qty(f(S)\cup g(S))$. 
Therefore, there exist a function $v: \mathrm{conv}\, \qty(f(S)\cup g(S)) \to \mathbb{R}$ and $(a, b) \in \mathbb{R}^n_{++}\times \mathbb{R}^n$ such that for all $i\in N$ and all $x\in \mathrm{conv}\, \qty(f(S)\cup g(S))$, $u_i (x) = a_i v(x) + b_i$. 

By  $f\not\succsim_i g$ for all $i\in N$, there exists $(p_1, \ldots, p_n ) \in \prod_{i\in N} P_i$ such that $\sum_{s\in S} p_i(s) u_i(g(s)) > \sum_{s\in S} p_i(s) u_i(f(s))$,  
which can be rewritten as for all $i\in N$, 
\begin{equation}
\label{eq:nece}
     \sum_{s\in S} p_i(s) v(g(s)) > \sum_{s\in S} p_i(s) v(f(s)). 
\end{equation}
By statement (ii), there exists $p_0\in P_0$ such that $p_0 = \sum_{i\in N}\gamma_i p_i$ for some $\gamma\in \Delta(N)$. 
Then, by \eqref{eq:nece}, 
\begin{align}
     \sum_{s\in S} p_0(s)  u_0(g(s))
     &=  \sum_{s\in S} \qty( \sum_{i\in N}\gamma_i p_i (s) ) \qty( \sum_{i\in N} \alpha_i ( a_i v(g(s)) +b_i ) +\beta) \\
     &
     >  \sum_{s\in S} \qty( \sum_{i\in N}\gamma_i p_i (s) ) \qty( \sum_{i\in N} \alpha_i ( a_i v(f(s)) +b_i ) +\beta) \\
     &> \sum_{s\in S} p_0(s) u_0 (f(s)), 
\end{align}
which implies $f\not\succsim_0 g$. 

%%%%%%%%%%%%%%%%%%
\subsection{Proof of Proposition \ref{prop:ex}}

    First, we prove that (i) implies (ii). 
    Let $f,g\in \mathcal{F}$ be such that they have no taste disagreement and $f \succsim_i g$ for all $i\in N$. 
    Then, there exist a function $v: \mathrm{conv}\, \qty(f(S)\cup g(S)) \to \mathbb{R}$ and $(a, b) \in \mathbb{R}^n_{++}\times \mathbb{R}^n$ such that for all $i\in N$ and all $x\in \mathrm{conv}\, \qty(f(S)\cup g(S))$, $u_i (x) = a_i v(x) + b_i$. 
    By  $f \succsim_i g$ for all $i\in N$, we have that for all $p\in \bigcup_{i\in N}P_i$, 
    \begin{equation}
        \sum_{s\in S} p (s) v (f(s)) \geq \sum_{s\in S} p (s) v (g(s)),
    \end{equation}
    which implies  $\sum_{s\in S}p_j(s)u_i(f(s))\geq \sum_{s\in S}p_j(s)u_i(g(s))$ for all $i,j\in N$ and all $p_j\in P_j$. Therefore, Exchange Pareto is stronger than Common-Taste Pareto. 
    By applying Theorem 2 of \citet{danan2016robust},  statement (ii) holds.

    Then, we prove the converse. Suppose that for all $i,j\in N$ and all $p_j\in P_j$,
    \begin{equation}
    \label{eq:agree}
        \sum_{s\in S} p_j (s) u_i(f(s)) \geq \sum_{s\in S} p_j (s) u_i(g(s)). 
    \end{equation}
    Let $p_0 \in P_0$. 
    By  (ii), there exist $(p_1,\ldots, p_n) \in \prod_{i\in N} P_i$ and $\gamma\in \Delta(N)$ such that $p_0 = \sum_{i\in N} \gamma_i p_i$.
    Then, by \eqref{eq:agree}, 
    \begin{align*}
        \sum_{s\in S} p_0(s) u_0(f(s))
        &= \sum_{s\in S} \qty( \sum_{j\in N} \gamma_j p_j(s)) \qty(\sum_{i\in N} \alpha_i u_i (f(s))+ \beta  ) = \sum_{i,j\in N} \alpha_i \gamma_j \qty(\sum_{s\in S} p_j(s) u_i(f(s))) + \beta\\
        &\geq \sum_{i,j\in N} \alpha_i \gamma_j \qty(\sum_{s\in S} p_j(s) u_i(g(s))) + \beta = \sum_{s\in S} \qty( \sum_{j\in N} \gamma_j p_j(s)) \qty(\sum_{i\in N} \alpha_i u_i (g(s)) + \beta  )\\
        &= \sum_{s\in S} p_0(s) u_0(g(s)).  
    \end{align*}
    Since this holds for each $p_0\in P_0$, we have $f\succsim_0 g$.

%%%%%%%%%%%%%%%%%%
\subsection{Proof of Proposition \ref{prop:exstar}}

    First, we prove that (i) implies (ii).   
    Let $f,g\in \mathcal{F}$ be such that they have no taste disagreement and $f \not\succsim_i g$ for all $i\in N$.
    Then, there exist a function $v: \mathrm{conv}\, \qty(f(S)\cup g(S)) \to \mathbb{R}$ and $(a, b) \in \mathbb{R}^n_{++}\times \mathbb{R}^n$ such that for all $i\in N$ and all $x\in \mathrm{conv}\, \qty(f(S)\cup g(S))$, $u_i (x) = a_i v(x) + b_i$. 
    For each $j\in N$, it follows from $f \not\succsim_j g$ that there exists $p_j\in P_j$ such that 
    \begin{equation}
        \sum_{s\in S} p_j (s) v (g(s)) > \sum_{s\in S} p_j (s) v (f(s)) ,
    \end{equation}
    which implies $\sum_{s\in S} p_j(s) u_i(g(s)) > \sum_{s\in S} p_j (s) u_i(f(s))$ for all $i\in N$. Therefore, Exchange Pareto{\pstar} is stronger than Common-Taste Pareto{\pstar}. 
    By applying Theorem \ref{thm:common}, statement (ii) holds.
    
    We then prove the converse.
    Suppose that  there exists $(p_1,\ldots, p_n) \in \prod_{j\in N} P_j$ such that  $\sum_{s\in S} p_j(s) u_i(g(s)) > \sum_{s\in S} p_j (s) u_i(f(s))$ for all $i,j\in N$. 
    By (ii), there exists $p_0 \in P_0$ such that $p_0 = \sum_{i\in N} \gamma_i p_i$ for some $\gamma\in \Delta(N)$. Then, since there exist $(\alpha, \beta)\in\qty(\mathbb{R}_+^n \setminus \{ \mathbf{0} \}) \times \mathbb{R}$ such that $u_0 = \alpha_i u_i+\beta$, we have 
    \begin{align}
        \sum_{s\in S} p_0(s) u_0(f(s))
        &= \sum_{s\in S} \qty( \sum_{j\in N} \gamma_j p_j(s)) \qty(\sum_{i\in N} \alpha_i u_i (f(s))+ \beta  ) = \sum_{i,j\in N} \alpha_i \gamma_j \qty(\sum_{s\in S} p_j(s) u_i(f(s))) + \beta\\
        &< \sum_{i,j\in N} \alpha_i \gamma_j \qty(\sum_{s\in S} p_j(s) u_i(g(s))) + \beta = \sum_{s\in S} \qty( \sum_{j\in N} \gamma_j p_j(s)) \qty(\sum_{i\in N} \alpha_i u_i (g(s)) + \beta  )\\
        &= \sum_{s\in S} p_0(s) u_0(g(s)), 
    \end{align}
    which implies $f\succsim_0 g$. 
    
%%%%%%%%%%%%%%%%%%%%%%%%%%%%%%%%%%%%
%%%%%%%%%%%%%%%%%%%%%%%%%%%%%%%%%%%%
\bibliographystyle{econ-aea}
\bibliography{reference}
%%%%%%%%%%%%%%%%%%%%%%%%%%%%%%%%%%%%
%%%%%%%%%%%%%%%%%%%%%%%%%%%%%%%%%%%%

\end{document}